\documentclass[11pt,letterpaper]{article}

\usepackage{fullpage}
\usepackage{amsthm}
\usepackage{amssymb}
\usepackage{amsmath}
\usepackage{multirow}
\usepackage{tabularx}
\usepackage{colortbl}
\usepackage{color}
\usepackage{graphicx}
\usepackage{url}
\usepackage[hypertex]{hyperref}
\usepackage{hyperref}
\usepackage{ifthen}

\newtheorem{theorem}{Theorem}[section]
\newtheorem{lemma}[theorem]{Lemma}

\newtheorem{fact}[theorem]{Fact}

\newcommand{\zo}{\{0,1\}}

\newcommand{\N}{{\mathbb{N}}}
\newcommand{\Z}{{\mathbb{Z}}}
\newcommand{\R}{{\mathbb{R}}}

\newcommand{\eps}{\epsilon}
\newcommand{\tO}{\tilde{O}}

\newcommand{\E}[2]{{\mathbb{E}_{#1}\left[#2\right]}}

\DeclareMathOperator{\ed}{ed}

\newcommand{\M}{{\cal M}}

\DeclareMathOperator{\TEMD}{TEMD}

\DeclareMathOperator{\TM}{TM}

\newcommand{\prelim}{\footnote{A preliminary version of this paper appeared in {\sl Proceedings of the 41st Annual ACM Symposium on Theory of Computing (STOC 2009)}, Bethesda, MD, USA, 2009, pp.\ 199--204.}}

\title{Approximating Edit Distance in Near-Linear Time\protect\prelim}

\author{
Alexandr Andoni\thanks{This work was done when the author was at
Massachusetts Institute of Technology, while supported in
  part by David and Lucille Packard Fellowship and by
MADALGO (Center for Massive Data Algorithmics, funded by the Danish
National Research Association) and by NSF grant CCF-0728645.
}\\
Microsoft Research SVC
\and
Krzysztof Onak\thanks{Supported in part by a Symantec research fellowship, NSF grant 0728645, and NSF grant 0732334. This work was done when the author was a graduate student at Massachusetts Institute of Technology.}\\
Carnegie Mellon University
}

\begin{document}

\maketitle

\begin{abstract}
We show how to compute the edit distance between two strings of length
$n$ up to a factor of $2^{\tO(\sqrt{\log n})}$ in $n^{1+o(1)}$
time. This is the first sub-polynomial approximation algorithm for this problem
that runs in near-linear time, improving on the state-of-the-art
$n^{1/3+o(1)}$ approximation. Previously, approximation of
$2^{\tO(\sqrt{\log n})}$ was known only for {\em embedding} edit
distance into $\ell_1$, and it is not known if that
embedding can be computed in less than quadratic time.
\end{abstract}

\section{Introduction}
The {\em edit distance} (or {\em Levenshtein distance}) between two
strings is the number of insertions, deletions, and substitutions
needed to transform one string into the other~\cite{Lev65}. This
distance is of fundamental importance in several fields such as
computational biology and text processing/searching, and consequently,
problems involving edit distance were studied extensively
(see~\cite{Navarro01}, \cite{Gus-book}, and references therein).  In
computational biology, for instance, edit distance and its slight
variants are the most elementary measures of dissimilarity for
genomic data, and thus improvements on edit distance algorithms have
the potential of major impact.

The basic problem is to compute the edit distance between two strings
of length $n$ over some alphabet. The text-book dynamic programming
runs in $O(n^2)$ time (see~\cite{CLRS} and references therein). This
was only slightly improved by Masek and Paterson~\cite{MP80} to
$O(n^2/\log^2 n)$ time for constant-size alphabets\footnote{The result
  has been only recently extended to arbitrarily large alphabets by
  Bille and Farach-Colton~\cite{BFC08} with a $O(\log\log n)^2$
  factor loss in time.}. Their result from 1980 remains the best
algorithm to this date.

Since near-quadratic time is too costly when working on large
datasets, practitioners tend to rely on faster heuristics
(see~\cite{Gus-book}, \cite{Navarro01}). This leads to the
question of finding fast algorithms with provable guarantees,
specifically: can one {\em approximate} the edit distance between two
strings in near-linear time \cite{I-survey, BEK+03, BJKK04, BES06,
  CPSV, C-PhD, OR-edit, KN, KR06}\,?

\newcommand{\localfootnote}{\footnote{We make no
    attempt at presenting a complete list of results for restricted
    problems, such as average case edit distance, weakly-repetitive
    strings, bounded distance regime, or related problems, such as
    pattern matching/nearest neighbor, sketching. However, for a very
    thorough survey, if only slightly outdated,
    see~\cite{Navarro01}.}}
\paragraph{Prior results on approximate algorithms\protect\localfootnote.}
A linear-time $\sqrt{n}$-approximation algorithm
immediately follows from the $O(n+d^2)$-time exact algorithm
(see Landau, Myers, and Schmidt \cite{LMS98}), where $d$ is the edit
distance between the input strings. Subsequent research improved the
approximation first to $n^{3/7}$, and then to $n^{1/3+o(1)}$,
due to, respectively, Bar-Yossef, Jayram, Krauthgamer, and
Kumar~\cite{BJKK04}, and Batu, Erg\"{u}n, and Sahinalp~\cite{BES06}.

A {\em sublinear} time algorithm was obtained by Batu, Erg\"{u}n,
Kilian, Magen, Raskhodnikova, Rubinfeld, and Sami~\cite{BEK+03}. Their
algorithm distinguishes the cases when the distance is $O(n^{1-\eps})$
vs.\ $\Omega(n)$ in $\tO(n^{1-2\eps}+n^{(1-\eps)/{2}})$ time\footnote{We
  use $\tO(f(n))$ to denote $f(n)\cdot \log^{O(1)} f(n)$.}
for any $\eps>0$.
Note that their algorithm cannot distinguish distances, say, $O(n^{0.1})$
vs.\ $\Omega(n^{0.9})$.

On a related front, in 2005, the breakthrough result of Ostrovsky and
Rabani gave an {\em embedding} of the edit distance metric into
$\ell_1$ with $2^{\tO(\sqrt{\log n})}$ distortion~\cite{OR-edit} (see
preliminaries for definitions). This result vastly improved related
applications, namely nearest neighbor search and sketching. However,
it did not have implications for computing edit distance between two
strings in sub-quadratic time. In particular, to the best of our
knowledge it is not known whether it is possible to compute their
embedding in less than quadratic time.

The best approximation to this date remains the 2006 result of Batu,
Erg\"{u}n, and Sahinalp~\cite{BES06}, achieving $n^{1/3+o(1)}$
approximation. Even for $n^{2-\eps}$ time, their approximation is
$n^{\eps/3+o(1)}$.

\paragraph{Our result.} We obtain $2^{\tO(\sqrt{\log n})}$ approximation in
near-linear time. This is the first sub-polynomial approximation
algorithm for computing the edit distance between two strings running
in strongly subquadratic time.

\begin{theorem}
\label{thm:main}
The edit distance between two strings $x,y\in\zo^n$ can be computed up
to a factor of $2^{O(\sqrt{\log n\log\log n})}$ in $n\cdot
2^{O(\sqrt{\log n\log\log n})}$ time.
\end{theorem}

Our result immediately extends to two more related applications.  The
first application is to sublinear-time algorithms. In this scenario,
the goal is to compute the distance between two strings $x,y$ of the
same length $n$ in $o(n)$ time. For this problem, for any
$\alpha<\beta \le 1$, we can distinguish distance $O(n^\alpha)$ from
distance $\Omega(n^\beta)$ in $O(n^{\alpha+2(1-\beta)+o(1)})$ time.

The second application is to the problem of pattern matching with
errors. In this application, one is given a text $T$ of length $N$ and
a pattern $P$ of length $n$, and the goal is to report the substring
of $T$ that minimizes the edit distance to $P$. Our result immediately
gives an algorithm for this problem running in $O(N \log N)\cdot
2^{\tilde O(\sqrt{\log n})}$ time with $2^{\tilde O(\sqrt{\log n})}$
approximation. We note that the best {\em exact} algorithm for this
problem runs in time $O(Nn/\log^2 n)$ \cite{MP80}. Better algorithms
may be obtained if we restrict the minimal distance between the
pattern and best substring of $T$ or for relatives of the edit
distance. In particular, Sahinalp and Vishkin~\cite{SV96} and Cole and
Hariharan~\cite{CH-approx} showed linear-time algorithms for finding
all substrings at distance at most $n^{c}$, where $c$ is a constant in
$(0,1)$. Moreover, Cormode and Muthukrishnan gave a near-linear time
$\tO(\log n)$-approximation algorithm when the distance is the
\emph{edit distance with moves}.

\subsection{Preliminaries and Notation}

Before describing our general approach and the techniques used, we
first introduce a few definitions.

We write $\ed(x,y)$ to denote the edit distance between strings $x$ and $y$.
We use the notation $[n]=\{1,2,3,\ldots n\}$. For a string $x$, a
substring starting at $i$, of length $m$, is denoted
$x[i:i+m-1]$. Whenever we say \emph{with high probability} (w.h.p.)
throughout the paper, we mean ``with probability $1-1/p(n)$'', where
$p(n)$ is a sufficiently large polynomial function of the input size $n$.

\paragraph{Embeddings.} For a metric $(M,d_M)$, and another metric
$(X,\rho)$, an {\em embedding} is a map $\phi:M\to X$ such that, for
all $x,y\in M$, we have $ d_M(x,y)\le\rho(\phi(x),\phi(y))\le
\gamma\cdot d_M(x,y) $ where $\gamma\ge 1$ is the {\em distortion} of
the embedding. In particular, all embeddings in this paper are
non-contracting.

We say embedding $\phi$ is {\em oblivious} if for any
subset $S\subset M$ of size $n$, the distortion guarantee holds for all pairs
$x,y\in S$ with high probability. The embedding $\phi$ is {\em
  non-oblivious} if it holds for a specific set $S$ (i.e., $\phi$ is
allowed to depend on $S$).

\paragraph{Metrics.} The $k$-dimensional $\ell_1$ metric is the set of points living in
$\R^k$ under the distance $\|x-y\|_1=\sum_{i=1}^k |x_i-y_i|$. We also
denote it by $\ell_1^k$.

We define {\em thresholded Earth-Mover Distance}, denoted $\TEMD_t$
for a fixed threshold $t>0$, as the following distance on subsets $A$
and $B$ of size $s\in \N$ of some metric $(M,d_M)$:
\begin{equation}
\label{eqn:temd}
\TEMD_t(A,B)=\tfrac{1}{s}\min_{\tau:A\to B}\sum_{a\in A}
\min\big\{d_M(a,\tau(a)), t\big\}
\end{equation}
where $\tau$ ranges over all bijections between sets $A$ and
$B$. $\TEMD_{\infty}$ is the simple Earth-Mover Distance (EMD).  We
will always use $t=s$ and thus drop the subscript $t$; i.e.,
$\TEMD=\TEMD_s$.

A {\em graph (tree) metric} is a metric induced by a connected
weighted graph (tree) $G$, where the distance between two vertices is
the length of the shortest path between them. We denote an
arbitrary tree metric by $\TM$.

\paragraph{Semimetric spaces.} We define a semimetric
to be a pair $(M,d_M)$ that satisfies all the properties of a metric space
except the triangle inequality. A {\em $\gamma$-near metric} is a
semimetric $(M,d_M)$ such that there exists some metric $(M,d_M^*)$
(satisfying the triangle inequality) with the property that, for any
$x,y\in M$, we have that
$
d_M^*(x,y)\le d_M(x,y) \le \gamma\cdot d_M^*(x,y).
$

\paragraph{Product spaces.} A {\em sum-product over a metric $\M=(M,d_M)$}, denoted
$\bigoplus_{\ell_1}^k \M$, is a derived metric over the set $M^k$,
where the distance between two points $x=(x_1,\ldots x_k)$ and
$y=(y_1,\ldots y_k)$ is equal to
$$
d_{1,M}(x,y)=\sum_{i\in[k]} d_M(x_i,y_i).
$$

For example the space $\bigoplus_{\ell_1}^k \R$ is just the
$k$-dimensional $\ell_1$.

Analogously, a {\em min-product over $\M=(M,d_M)$}, denoted
$\bigoplus_{\min}^k \M$, is a semimetric over $M^k$, where
the distance between two points $x=(x_1,\ldots x_k)$ and
$y=(y_1,\ldots y_k)$ is
$$
d_{\min,M}(x,y)=\min_{i\in[k]}\big\{d_M(x_i,y_i)\big\}.
$$

We also slightly abuse the notation by writing $\bigoplus_{\min}^k
\TM$ to denote the min-product of $k$ tree metrics (that could differ
from each other).

\subsection{Techniques}
Our starting point is the Ostrovsky-Rabani
embedding~\cite{OR-edit}. For strings $x,y$, as well as for all
substrings $\sigma$ of specific lengths, we compute some vectors
$v_\sigma$ living in low-dimensional $\ell_1$ such that the distance
between two such vectors approximates the edit distance between the
associated (sub-)strings. In this respect, these vectors can be seen
as an embedding of the considered strings into $\ell_1$ of {\em
  polylogarithmic dimension}. Unlike the Ostrovsky-Rabani embedding,
however, our embedding is {\em non-oblivious} in the sense that the
vectors $v_\sigma$ are computed given all the relevant strings
$\sigma$. In contrast, Ostrovsky and Rabani give an {\em oblivious}
embedding $\phi_n:\zo^n\to \ell_1$ such that
$\|\phi_n(x)-\phi_n(y)\|_1$ approximates $\ed(x,y)$. However, the
obliviousness comes at a high price: their embedding requires a high
dimension, of order $\Omega(n)$, and a high computation time, of order
$\Omega(n^2)$ (even when allowing randomized embedding, and a constant
probability of a correctness). We further note that reducing the
dimension of this embedding seems unlikely as suggested by the results
on impossibility of dimensionality reduction within
$\ell_1$~\cite{CS-linearL1, BC, LN}. Nevertheless, the general
recursive approach of the Ostrovsky-Rabani embedding is the starting
point of the algorithm from this paper.

The heart of our algorithm is a near-linear time algorithm that, given
a sequence of low-dimensional vectors $v_1,\ldots v_n\in \ell_1$ and
an integer $s<n$, constructs new vectors $q_1,\ldots q_m\in
\ell_1^{O(\log^2 n)}$, where $m=n-s+1$, with the following
property. For all $i,j\in[m]$, the value $\|q_i-q_j\|_1$ approximates
the Earth-Mover Distance (EMD)\footnote{In fact, our algorithm does
  this for thresholded EMD, TEMD, but the technique is precisely the
  same.}  between the sets $A_i=\{v_i,v_{i+1},\ldots v_{i+s-1}\}$ and
$A_j=\{v_j,v_{j+1},\ldots v_{j+s-1}\}$. To accomplish this
(non-oblivious) embedding, we proceed in two stages. First, we embed
(obliviously) the EMD metric into a {\em min-product of $\ell_1$'s}
of low dimension. In other words, for a set $A$, we associate a matrix
$L(A)$, of polylogarithmic size, such that the EMD distance between
sets $A$ and $B$ is approximated by $\min_r\sum_t
|L(A)_{rt}-L(B)_{rt}|$. Min-products help us simultaneously on two
fronts: one is that we can apply a {\em weak} dimensionality reduction
in $\ell_1$, using the Cauchy projections, and the second one enables
us to accomplish a low-dimensional EMD embedding itself. Our embedding
$L(\cdot)$ is not only low-dimensional, but it is also {\em linear},
allowing us to compute matrices $L(A_i)$ in near-linear time by
performing one pass over the sequence $v_1,\ldots v_n$. Linearity is
crucial here as even the total size of $A_i$'s is $\sum_i
|A_i|=(n-s+1)\cdot s$, which can be as high as $\Omega(n^2)$, and so
processing each $A_i$ separately is infeasible.

In the second stage, we show how to embed a set of $n$ points lying in
a low-dimensional min-product of $\ell_1$'s back into a
low-dimensional $\ell_1$ with only small distortion. We note that
this is not possible in general, with any bounded distortion, because such a
set of points does not even form a metric. We show that this is
possible when we assume that the semi-metric induced by the set of
points approximates some metric (in our case, the set of
points approximates the initial EMD metric). The embedding from this stage
starts by embedding a min-product of $\ell_1$'s into a low-dimensional
min-product of tree metrics. We further embed the latter into an
$n$-point metric supported by the shortest-path metric of a {\em sparse}
graph. Finally, we observe that we can implement Bourgain's embedding
on a sparse graph metric in {\em near-linear time}. These last two
steps make our embedding non-oblivious.

\subsection{Recent Work}

We note that the recent work \cite{AKO-edit} has shown that one can
approximate the edit distance between two strings up to a multiplicative factor of
$(\log n)^{O(1/\eps)}$ in $n^{1+\eps}$ time, for any desired $\eps>0$. Although the new
result obtains polylogarithmic approximation, the running time is slightly
higher than the algorithm presented here. For a comparable
approximation, obtained for $\eps=\sqrt{\log\log n / \log n}$, the
algorithm of \cite{AKO-edit} does not improve the running time (up to
constants hidden by the big O notation). We further remark that the techniques of
\cite{AKO-edit} are disjoint from the techniques presented here, and
are based on asymmetric sampling of one of the strings.

\section{Short Overview of the Ostrovsky-Rabani Embedding}
\label{sec:ORreview}

We now briefly describe the embedding of Ostrovsky and
Rabani~\cite{OR-edit}.
Some notions introduced here are used in our algorithm described in
the next section.

The embedding of Ostrovsky and Rabani is recursive.  For a fixed $n$, they
construct the embedding of edit distance over strings of length $n$
using the embedding of edit distance over strings of shorter lengths
$l\le n/2^{\sqrt{\log n\log\log n}}$. 
We denote their embedding of length-$n$ strings by
$\phi_n:\zo^n\to\ell_1$, and let $d_n^{\rm OR}$ be the resulting
distance: $d_n^{\rm OR}(x,y)=\|\phi_n(x)-\phi_n(y)\|_1$.  For two
strings $x,y\in\zo^n$, the embedding is such that $d_n^{\rm
  OR}=\|\phi_n(x)-\phi_n(y)\|_1$ approximates an ``idealized''
distance $d^*_n(x,y)$, which itself approximates the edit distance
between $x$ and $y$.

Before describing the ``idealized'' distance $d^*_n$, we introduce
some notation. Partition $x$ into $b=2^{\sqrt{\log n\log\log n}}$
blocks called $x^{(1)},\ldots x^{(b)}$ of length $l=n/b$. Next, fix
some $j\in[b]$ and $s\le l$. We consider the set of all substrings of
$x^{(j)}$ of length $l-s+1$, embed each one recursively via
$\phi_{l-s+1}$, and define $S_j^s(x)\subset \ell_1$ to be the set of
resulting vectors (note that $|S_j^s|=s$). Formally,
$$ S_j^s(x)=\big\{\phi_{l-s+1}(x[(j-1)l+z:(j-1)l+z+l-s])\mid z\in[s]
\big\}.$$ Taking $\phi_{l-s+1}$ as given (and thus also the sets
$S_j^s(x)$ for all $x$), define the new ``idealized'' distance
$d^*_n$ approximating the edit distance between strings $x,y\in\zo^n$
as
\begin{equation}
\label{eqn:idealDist}
d^*_n(x,y)=c\sum_{j=1}^{b} \sum_{\stackrel{f\in \N}{s=2^f\le l}}
\TEMD(S_j^s(x), S_j^s(y))
\end{equation}
where TEMD is the thresholded Earth-Mover Distance (defined in
Equation~\eqref{eqn:temd}), and $c$ is a sufficiently large normalization
constant ($c\ge 12$ suffices).  Using the terminology from the
preliminaries, the distance function $d^*_n$ can be viewed as the
distance function of the sum-product of TEMDs, i.e.,
$\bigoplus_{\ell_1}^{b}\bigoplus_{\ell_1}^{O(\log n)}\TEMD$, and the
embedding into this product space is attained by the natural identity
map (on sets $S_j^s$).

The key idea is that the distance $d_n^*(x,y)$ approximates edit
distance well, assuming that $\phi_{l-s+1}$ approximates edit distance
well, for all $s=2^f$ where $f\in \{1,2,\ldots \lfloor \log_2
l\rfloor\}$. Formally, Ostrovsky and Rabani show that:
\begin{fact}[\cite{OR-edit}]
\label{fct:ORrecursion}
Fix $n$ and $b<n$, and let $l=n/b$.  Let $D_{n/b}$ be an upper bound on
distortion of $\phi_{l-s+1}$ viewed as an embedding of edit distance
on strings $\{x[i:i+l-s], y[i:i+l-s]\mid i\in[n-l+s]\}$, for all
$s=2^f$ where $f\in \{1,2,\ldots \lfloor \log_2 l\rfloor\}$.  Then,
$$
\ed(x,y)\le d^*_n(x,y) \le \ed(x,y)\cdot \left(D_{n/b}+b\right)\cdot O(\log n).
$$
\end{fact}

To obtain a complete embedding, it remains to construct an embedding
approximating $d^*_n$ up to a small factor. In fact, if one manages to
approximate $d^*_n$ up to a poly-logarithmic factor, then the final
distortion comes out to be $2^{O(\sqrt{\log n\log\log n})}$. This
follows from the following recurrence on the distortion
factor $D_n$. Suppose $\phi_n$ is an embedding that approximates $d^*_n$ up
to a factor $\log^{O(1)} n$. Then, if $D_n$ is the distortion of
$\phi_n$ (as an embedding of edit distance), then
Fact~\ref{fct:ORrecursion} immediately implies that, for
$b=2^{\sqrt{\log n\log\log n}}$,
$$ D_n\le D_{n/2^{\sqrt{\log n\log\log n}}}\cdot \log^{O(1)}
n+2^{O(\sqrt{\log n\log\log n})}.
$$ This recurrence solves to $D_n\le 2^{O(\sqrt{\log n\log\log n})}$
as proven in~\cite{OR-edit}.

Concluding, to complete a step of the recursion, it is sufficient to
embed the metric given by $d^*_n$ into $\ell_1$ with a polylogarithmic
distortion. Recall that $d^*_n$ is the distance of the metric
$\bigoplus_{\ell_1}^{b}\bigoplus_{\ell_1}^{O(\log n)} \TEMD$, and thus,
one just needs to embed $\TEMD$ into $\ell_1$.  Indeed, Ostrovsky and
Rabani show how to embed a relaxed (but sufficient) version of TEMD
into $\ell_1$ with $O(\log n)$ distortion, yielding the desired
embedding $\phi_n$, which approximates $d_n^*$ up to a $O(\log n)$
factor at each level of recursion. We note that the required dimension
is $\tO(n)$.

\section{Proof of the Main Theorem}

We now describe our general approach. Fix $x\in\zo^n$. For
each substring $\sigma$ of $x$, we construct a low-dimensional
vector $v_\sigma$ such that, for any two substrings $\sigma,\tau$ of the
same length, the edit distance between $\sigma$ and $\tau$ is
approximated by the $\ell_1$ distance between the vectors $v_\sigma$
and $v_\tau$. We note that the embedding is non-oblivious: to
construct vectors $v_\sigma$ we need to know {\em all} the
substrings of $x$ in advance (akin to Bourgain's embedding
guarantee). We also note that computing such vectors is enough to
solve the problem of approximating the edit distance between two
strings, $x$ and $y$. Specifically, we apply this procedure to the
string $x'=x\circ y$, the concatenation of $x$ and $y$, and then
compute the $\ell_1$ distance between the vectors corresponding to $x$
and $y$, substrings of $x'$.

More precisely, for each length $m\in W$, for some set $W\subset [n]$
specified later, and for each substring $x[i:i+m-1]$, where
$i=1,\ldots n-m+1$, we compute a vector $v_i^{(m)}$ in $\ell_1^\alpha$, where
$\alpha=2^{\tO(\sqrt{\log n})}$. The construction is inductive: to compute vectors
$v_i^{(m)}$, we use vectors $v_i^{(l)}$ for $l\ll m$ and $l\in W$. The
general approach of our construction is based on the analysis of the
recursive step of Ostrovsky and Rabani, described in
Section~\ref{sec:ORreview}. In particular, our vectors
$v_i^{(m)}\in\ell_1$ will also approximate the $d^*_m$ distance (given
in Equation~\eqref{eqn:idealDist}) with sets $S_i^s$ defined using vectors
$v_i^{(l)}$ with $l\ll m$.  

The main challenge is to process one level (vectors $v_i^{(m)}$ for a
fixed $m$) in near-linear time. Besides the computation time itself, a
fundamental difficulty in applying the approach of Ostrovsky and
Rabani directly is that their embedding would give a much higher
dimension $\alpha$, proportional to $\tO(m)$.
Thus, if we were to use their embedding,
even storing all the vectors would take quadratic space.

To overcome this last difficulty, we settle on non-obliviously
embedding the set of substrings $x[i:i+m-1]$ for $i\in[n-m+1]$ under
the ``ideal'' distance $d^*_m$ with $\log^{O(1)} n$ distortion
(formally, under the distance $d^*_m$ from Equation~\eqref{eqn:idealDist},
when $S_j^s(x[i:i+m-1])=\left\{v_{i+(j-1)l+z-1}^{(l-s+1)}\mid
z\in[s]\right\}$ for $l=m/2^{\sqrt{\log n\log\log n}}$).
Existentially, we know that there exist vectors $w_i^{(m)}\in
\R^{O(\log^2 n)}$ such that $\|w_i^{(m)}-w_j^{(m)}\|_1$ approximates
$d^*_m(x[i:i+m-1],x[j:j+m-1])$ for all $i$ and $j$ --- this follows by the
standard Bourgain's embedding~\cite{Bou}. The vectors $v_i^{(m)}$ that we compute
approximate the properties of the ideal vectors $w_i^{(m)}$. Their efficient computability
comes at the cost of an additional polylogarithmic loss in approximation.

The main building block is the following theorem. It shows how to
approximate the TEMD distance for the desired sets $S_j^s$.

\medskip
\begin{theorem}
\label{thm:streamEMD}
Let $n\in \N$ and $s\in[n]$. Let $v_1,\ldots v_n$ be vectors in
$\{-M,\ldots M\}^\alpha$, where $M = n^{O(1)}$ and $\alpha\le
n$. Define sets $A_i=\{v_i,v_{i+1},\ldots v_{i+s-1}\}$ for $i \in
[n-s+1]$.

Let $t=O(\log^2 n)$. We can compute (randomized) vectors
$q_i\in\ell_1^t$ for $i\in[n-s+1]$ such that for any $i,j\in[n-s+1]$,
with high probability, we have
$$
\TEMD(A_i,A_j)\le \|q_i-q_j\|_1\le \TEMD(A_i,A_j)\cdot \log^{O(1)} n.
$$
Furthermore, computing all vectors $q_i$ takes $\tO(n\alpha)$ time.
\end{theorem}

To map the statement of this theorem to the above description, we
mention that, for each $l=m/b$ for $m\in W$, we apply
the theorem to vectors $\left(v_i^{(l-s+1)}\right)_{i\in [n-l+s]}$ for
each $s=1,2,4,8,\ldots 2^{\lfloor \log_2 l\rfloor}$.

We prove Theorem~\ref{thm:streamEMD} in later sections. Once
we have Theorem~\ref{thm:streamEMD}, it becomes relatively
straight-forward (albeit a bit technical) to prove the main theorem,
Theorem~\ref{thm:main}. We complete the proof of
Theorem~\ref{thm:main} next, assuming Theorem~\ref{thm:streamEMD}.

\begin{proof}[of Theorem~\ref{thm:main}]
We start by appending $y$ to the end of $x$; we will work with the
new version of $x$ only. Let $b=2^{\sqrt{\log n\log\log n}}$ and
$\alpha=O(b\log^3 n)$.  We construct vectors $v_i^{(m)}\in \R^\alpha$ for
$m\in W$, where $W\subset [n]$ is a carefully chosen set of size
$2^{O(\sqrt{\log n\log\log n})}$. Namely, $W$ is the minimal set such
that: $n\in W$, and, for each $i\in W$ with $i\ge b$, we have that
$i/b-2^j+1\in W$ for all integers $j\le \lfloor\log_2 i/b\rfloor$. It
is easy to show by induction that the size of $W$ is $2^{O(\sqrt{\log n\log\log n})}$.
We construct the vectors $v_i^{(m)}$ inductively in a bottom-up manner. We use vectors for small $m$ to build vectors for large $m$. $W$ is exactly the set of lengths $m$ that we need in the process.

Fix an $m\in W$ such that $m\le b^2=2^{2\sqrt{\log n\log\log n}}$. We
define the vector $v_i^{(m)}$ to be equal to $h_m(x[i:i+m-1])$, where
$h_m:\zo^m\to \zo^\alpha$ is a randomly chosen function. It is readily seen
that $\|v_i^{(m)}-v_j^{(m)}\|_1$ approximates $\ed(x[i:i+m-1],
x[j:j+m-1])$ up to $b^2=2^{2\sqrt{\log n\log\log n}}$ approximation
factor, for each $i,j\in[n-m+1]$.

Now consider $m\in W$ such that $m>b^2$. Let $l=m/b$. First we
construct vectors approximating TEMD on sets
$A_i^{m,s}=\left\{v_{i+z}^{(l-s+1)}\mid z=0,\ldots s-1\right\}$, where
$s=1,2,4,8,\ldots, l$ and $i\in[n-l+s]$. In particular, for a fixed
$s\in[l]$ equal to a power of 2, we apply Theorem~\ref{thm:streamEMD}
to the set of vectors $\left(v_i^{(l-s+1)}\right)_{i\in[n-l+s]}$
obtaining vectors
$\left(q_i^{(m,s)}\right)_{i\in[n-l+1]}$. Theorem~\ref{thm:streamEMD}
guarantees that, for each $i,j\in[n-l+1]$, the value
$\|q_i^{(m,s)}-q_j^{(m,s)}\|_1$ approximates
$\TEMD(A_i^{m,s},A_j^{m,s})$ up to a factor of $\log^{O(1)} n$. We can
then use these vectors $q_i^{(m,s)}$ to obtain the vectors
$v_i^{(m)}\in \R^{\alpha}$ that approximate the ``idealized'' distance
$d^*_m$ on substrings $x[i:i+m-1]$, for $i\in[n-m+1]$. Specifically,
we let the vector $v_i^{(m)}$ be a concatenation of vectors
$q_{i+(j-1)l}^{(m,s)}$, where $j\in[b]$, and $s$ goes over all powers of 2 less than
$l$:
$$
v_i^{(m)}=\Big( q_{i+(j-1)l}^{(m,s)}\Big)_{\stackrel{j\in[b]}{s=2^f\le
    l,f\in\N}}.
$$ Then, the vectors $v_i^{(m)}$ approximate the distance $d^*_m$
(given in Equation~\eqref{eqn:idealDist}) up to a~$\log^{O(1)}n$
approximation factor, with the sets $S_j^s(x[i:i+m-1])$ taken as
$$S_j^s(x[i:i+m-1])=A_{i+(j-1)l}^{m,s}
=\left\{v_{i+(j-1)l+z}^{(l-s+1)}\mid z=0,\ldots s-1\right\},$$ for
$i\in[n-m+1]$ and $j\in[b]$.

The algorithm finishes by outputting $\|v_{1}^{(n)}-v_{n+1}^{(n)}\|$,
which is an approximation to the edit distance between $x[1:n]$ and
$x[n+1:2n]=y$. The total running time is $O(|W|\cdot n\cdot
b^{O(1)}\cdot \log^{O(1)}n) = n\cdot 2^{O(\sqrt{\log n\log\log n})}$.

It remains to analyze the resulting approximation. Let $D_m$ be the
approximation achieved by vectors $v_i^{(k)}\in \ell_1$ for substrings
of $x$ of lengths $k$, where $k \in W$ and $k \le m$.
Then, using Fact~\ref{fct:ORrecursion} and the fact
that vectors $v_i^{(m)}\in \ell_1$ approximate $d^*_m$, we have that
$$
D_m\le \log^{O(1)} n\cdot \left(D_{m/b}+2^{\sqrt{\log n\log\log n}}\right).
$$

Since the total number of recursion levels is bounded by ${\log_b n}=\sqrt{\frac{\log
  n}{ \log\log n}}$, we deduce that $D_n=2^{O(\sqrt{\log n\log\log n})}$.
\end{proof}

\subsection{Proof of Theorem~\ref{thm:streamEMD}}

The proof proceeds in two stages. In the first stage we show an
embedding of the TEMD metric into a low-dimensional
space. Specifically, we show an (oblivious) embedding of TEMD into a
{\em min-product} of $\ell_1$. Recall that the min-product of $\ell_1$,
denoted $\bigoplus_{\min}^l\ell_1^k$, is a semi-metric where the
distance between two $l$-by-$k$ vectors $x,y\in\R^{l\times k}$ is
$d_{\min,1}(x,y)=\min_{i\in[l]}\left\{\sum_{j\in[k]}
|x_{i,j}-y_{i,j}|\right\}$. Our min-product of $\ell_1$'s has
dimensions $l=O(\log n)$ and $k=O(\log^3n)$. The min-product can be
seen as helping us on two fronts: one is the embedding of TEMD into
$\ell_1$ (of initially high-dimension), and another is a {\em weak}
dimensionality reduction in $\ell_1$, using Cauchy projections.  Both
of these embeddings are of the following form: consider a randomized
embedding $f$ into (standard) $\ell_1$ that has no contraction
(w.h.p.)  but the expansion is bounded only in the expectation (as
opposed to w.h.p.). To obtain a ``w.h.p.'' expansion, one standard
approach is to sample $f$ many times and concentrate the
expectation. This approach, however, will necessitate a high number of
samples of $f$, and thus yield a high final dimension. Instead,
the min-product allows us to take only $O(\log n)$ independent samples of
$f$.

We note that our embedding of TEMD into min-product of $\ell_1$,
denoted $\lambda$, is linear in the sets $A$: $\lambda(A)=\sum_{a\in
  A} \lambda(\{a\})$. The linearity allows us to compute the embedding
of sets $A_i$ in a streaming fashion: the embedding of $A_{i+1}$ is
obtained from the embedding of $A_i$ with $\log^{O(1)} n$ additional
processing. This stage appears in Section~\ref{sec:emd2min}.

In the second stage, we show that, given a set of $n$ points in
min-product of $\ell_1$'s, we can (non-obliviously) embed these points into
low-dimensional $\ell_1$ with $O(\log n)$ distortion. The time required is
near-linear in $n$ and the dimensions of the min-product of
$\ell_1$'s. 

To accomplish this step, we start by embedding the min-product of
$\ell_1$'s into a min-product of tree metrics.  Next, we show that $n$
points in the low-dimensional min-product of tree metrics can be
embedded into a graph metric supported by a {\em sparse} graph. We
note that this is in general not possible, with any (even
non-constant) distortion. We show that this is possible when we assume
that our subset of the min-product of tree metrics approximates some
actual metric (in our case, the min-product approximates the TEMD
metric). Finally, we observe that we can implement Bourgain's
embedding in near-linear time on a sparse graph metric. This stage
appears in Section~\ref{sec:min2l1}.

We conclude with the proof of Theorem~\ref{thm:streamEMD} in
Section~\ref{sec:pfStEMD}.

\subsubsection{Embedding EMD into min-product of $\ell_1$}
\label{sec:emd2min}

In the next lemma, we show how to embed TEMD into a min-product of
$\ell_1$'s of low dimension. Moreover, when the sets $A_i$ are
obtained from a sequence of vectors $v_1,\ldots v_n$, by taking
$A_i=\{v_i,\ldots v_{i+s-1}\}$, we can compute the embedding in
near-linear time.

\begin{lemma}
\label{lem:slidingWindow}
Fix $n,M\in \N$ and $s\in[n]$. Suppose we have $n$ vectors $v_1,\ldots
v_n$ in $\{-M,-M+1,\ldots,M\}^\alpha$ for some $\alpha\le n$. Consider the sets
$A_i=\{v_i,v_{i+1},\ldots v_{i+s-1}\}$, for $i\in[n-s+1]$.

Let $k=O(\log^3 n)$. We can compute (randomized) vectors
$q_i\in\ell_1^k$ for $i\in[n-s+1]$ such that, for any $i,j\in[n-s+1]$
we have that
\begin{itemize}
\item
$
\Pr\Big[\|q_i-q_j\|_1\le \TEMD(A_i,A_j)\cdot O(\log^2 n)\Big] \ge 0.1
$ and
\item
$\|q_i-q_j\|_1\ge \TEMD(A_i,A_j)$ w.h.p.  
\end{itemize}
The computation time
is $\tO(n\alpha)$.

Thus, we can embed the $\TEMD$ metric over sets $A_i$ into
$\bigoplus_{\min}^l \ell_1^k$, for $l=O(\log n)$, such that the
distortion is $O(\log^2 n)$ w.h.p. The computation time is $\tO(n\alpha)$.
\end{lemma}

\begin{proof}
First, we show how to embed TEMD metric over the sets $A_i$ into
$\ell_1$ of dimension $M^{O(\alpha)} \cdot O(\log n)$.
For this purpose, we use a slight
modification of the embedding of~\cite{AIK} (it can also be seen as a
strengthening of the TEMD embedding of Ostrovsky and
Rabani).

The embedding of~\cite{AIK} constructs $m=O(\log s)$ embeddings
$\psi_i$, each of dimension $h=M^{O(\alpha)}$, and then the final
embedding is just the concatenation $\psi=\psi_1\circ \psi_2\ldots
\circ \psi_m$. For $i=1,\ldots m$, we impose a randomly shifted grid
of side-length $R_i=2^{i-2}$. That is, let $\Delta_i=(\delta_{i,1},\ldots,\delta_{i,\alpha})$
be selected uniformly at random from $[0,1)^\alpha$. A specific vector $v_j$ falls into the cell $(c_1,\ldots,c_\alpha)$, where $c_t = \lfloor v_{j,t}/R_i + \delta_{i,t} \rfloor$ for $t = 1,\ldots,\alpha$.
Then $\psi_i$ has a coordinate for each
cell $(c_1,\ldots,c_\alpha)$, where $0 \le c_t \le 2M/R_i+1$ for $t=1,\ldots,\alpha$.
These are the only cells that can be non-empty, and there
is at most $(2M/R_i+1)^\alpha = M^{O(\alpha)}$ of them. The value of a specific coordinate, for a set $A$, equals the
number of vectors from $A$ falling into the corresponding cell
times $R_i$. Now, if
we scale $\psi$ up by a factor of $\Theta(\tfrac{1}{s}\log n)$, Theorem~3.1 from~\cite{AIK}%
\footnote{Note that Theorem~3.1 from~\cite{AIK} is stated for EMD, and here we are concerned with TEMD. Nevertheless, the whole statement still applies, because the side of the largest grid is bounded by $O(s)$ .}
says that the vectors $q_i'=\psi(A_i)$ satisfy the condition that, for
any $i,j\in[n-s+1]$, we have:
\begin{itemize}
\item
$\E{}{\|q'_i-q'_j\|_1}\le
\TEMD(A_i,A_j)\cdot O(\log^2 n)$ and
\item
$\|q'_i-q'_j\|_1\ge
\TEMD(A_i,A_j)$ w.h.p. 
\end{itemize}
Thus, the vectors $q'_i$ satisfy the promised properties except they
have a high dimension.

To reduce the dimension of $q_i'$'s, we apply a weak $\ell_1$
dimensionality reduction via 1-stable (Cauchy) projections. Namely, we
pick a random matrix $P$ of size $k=O(\log^3 n)$ by $mh = O(\log s) \cdot M^{O(\alpha)}$, the
dimension of $\psi$, where each entry is distributed according to the
Cauchy distribution, which has probability distribution function
$f(x)=\tfrac{1}{\pi}\cdot\tfrac{1}{1+x^2}$. Now define $q_i=P\cdot q'_i\in
\ell_1^{k}$. Standard properties of the $\ell_1$ dimensionality reduction
guarantee that the vectors $q_i$ satisfy the properties promised in
the lemma statement, after an appropriate rescaling (see Theorem 5
of~\cite{I00b} with $\eps=1/2$, $\gamma=1/6$, and $\delta=n^{-O(1)}$).

It remains to show that we can compute the vectors $q_i$ in $\tO(n\alpha)$
time. To this end, observe that the resulting embedding $P\cdot \psi(A)$
is linear, namely $P\cdot \psi(A)=\sum_{a\in A} P\cdot
\psi(\{a\})$. Moreover, each $P\cdot\psi(\{v_{i}\})$ can be computed in $\alpha\cdot
\log^{O(1)} n$ time, because $\psi(\{v_{i}\})$ has exactly one non-zero coordinate,
which can be computed in $O(\alpha)$ time, and then $P\cdot\psi(\{v_{i}\})$ is simply
the corresponding column of $P$ multiplied by the non-empty coordinate of $\psi(\{v_{i}\})$.
To obtain the first vector $q_1$, we compute the summation of all corresponding $P\cdot\psi(\{v_{i}\})$.
To compute the remaining vectors $q_i$ iteratively, we use the idea of a sliding window over the
sequence $v_1,\ldots v_n$. Specifically, we have
$$q_{i+1} = P\cdot \psi(A_{i+1})=P\cdot
\psi(A_i\cup\{v_{i+s}\}\setminus \{v_{i}\}) =
q_i+P\cdot\psi(\{v_{i+s}\})-P\cdot\psi(\{v_{i}\}),$$ 
which implies that $q_{i+1}$ can be computed in $\alpha\cdot
\log^{O(1)} n$ time, given the value of $q_{i}$.
Therefore, the total time required to compute
all $q_i$'s is $O(n\alpha\cdot \log^{O(1)} n)$.

Finally, we show how we obtain an efficient embedding of TEMD into
min-product of $\ell_1$'s.
We apply the above procedure $l=O(\log n)$ times. Let $q_i^{(z)}$ be
the resulting vectors, for $i\in[n-s+1]$ and $z\in[l]$. The embedding
of a set $A_i$ is the concatenation of the vectors $q_i^{(z)}$,
namely $Q_i=(q_i^{(1)}, q_i^{(2)},\ldots
q_i^{(l)})\in\bigoplus_{\min}^l\ell_1^k$. The Chernoff bound implies
that w.h.p., for any $i,j\in[n-s+1]$, we have that
$$
d_{\min,1}(Q_i,Q_j)=\min_{z\in[l]} \|q_i^{(z)}-q_j^{(z)}\|\le
\TEMD_s(A_i,A_j)\cdot O(\log^2 n).
$$
Also, $d_{\min,1}(Q_i,Q_j)\ge \TEMD_s(A_i,A_j)$ w.h.p. trivially. Thus the
vectors $Q_i$ are an embedding of the TEMD metric on $A_i$'s into
$\bigoplus_{\min}^l\ell_1^k$ with distortion $O(\log^2 n)$ w.h.p.
\end{proof}

\subsubsection{Embedding of min-product of $\ell_1$ into
  low-dimensional $\ell_1$}
\label{sec:min2l1}

In this section, we show that $n$ points $Q_1,\ldots Q_n$ in the
semi-metric space $\bigoplus_{\min}^l\ell_1^k$ can be embedded into
$\ell_1$ of dimension $O(\log^2 n)$ with distortion $\log^{O(1)} n$. The
embedding works under the assumption that the semi-metric on
$Q_1,\ldots Q_n$ is a $\log^{O(1)} n$ approximation of some metric. We
start by showing that we can embed a min-product of $\ell_1$'s into
a min-product of tree metrics.

\begin{lemma}
\label{lem:l1ToHST}
Fix $n,M\in\N$ such that $M = n^{O(1)}$.  Consider $n$ vectors
$v_1,\ldots v_n$ in $\bigoplus_{\min}^l \ell_1^k$, for some $l,k\in
\N$, where each coordinate of each $v_i$ lies in the set $\{-M,
\ldots, M\}$. We can embed these vectors into a min-product of
$O(l\cdot\log^2 n)$ tree metrics, i.e., $\bigoplus_{\min}^{O(l\log^2
  n)} \TM$, incurring distortion $O(\log n)$ w.h.p. The computation
time is $\tO(n\cdot kl)$.
\end{lemma}

\begin{proof}
We consider all thresholds $2^t$, for $t \in \{0,1,\ldots, \log M\}$.
For each threshold $2^t$, and for each coordinate of the min-product
(i.e., $\ell_1^k$), we create $O(\log n)$ tree metrics. Each tree metric is
independently created as follows.  We again use randomly shifted
grids. Specifically, we define a hash function $h:\ell_1^k \to \Z^k$ as
$$h(x_1,\ldots,x_k) = \left(
\left\lfloor\frac{x_1+u_1}{2^t}\right\rfloor,
\left\lfloor\frac{x_2+u_2}{2^t}\right\rfloor, \ldots,
\left\lfloor\frac{x_k+u_k}{2^t}\right\rfloor\right),$$ where each
$u_t$ is chosen at random from $[0,2^t)$.  We create each tree metric
  so that the nodes corresponding to the points hashed by $h$ to the
  same value are at distance $2^t$ (this creates a set of stars), and
  each pair of points that are hashed to different values are at
  distance $2Mk$ (we connect the roots of the stars).  

For two points $x,y\in\ell_1^k$, the probability that they are separated by the grid in the $i$-th dimension is at most $|x_i -y_i|/2^t$, which implies by the union bound that
$$\Pr_h[h(x)=h(y)] \ge 1 - \sum_i \frac{|x_i - y_i|}{2^t} = 1 -\frac{\|x-y\|_1}{2^t}.$$
On the other hand, the probability that $x$ and $y$ are not separated by the grid in the $i$-th dimension is $\max\{1 - |x_i - y_i|/2^t, 0\} \le e^{-|x_i - y_i|/2^t}$. Since the grid is shifted independently in each dimension,
$$\Pr_h[h(x)=h(y)] \le \prod_{i=1}^{k} e^{-|x_i - y_i|/2^t} = e^{-\sum_{i=1}^{k}|x_i - y_i|/2^t} = 
e^{-\|x-y\|_1/2^t}.$$

By the Chernoff bound, if $x,y \in \ell_1^k$
  are at distance at most $2^t$ for some $t$, they will be at distance
  at most $2^{t+1}$ in one of the tree metrics with high probability.
On the other hand, let $v_i$ and $v_j$ be two input vectors at
distance greater than $2^t$. The probability that they are at distance
smaller than $2^t / c\log n$ in any of the $O(\log^2 n)$ tree
metrics, is at most $n^{-c+1}$ for any $c>0$, by the union bound.

Therefore, we multiply the weights of all edges in all trees by $O(\log n)$ 
to achieve a proper (non-contracting) embedding. 
\end{proof}

We now show that we can embed a subset of the min-product of tree
metrics into a graph metric, assuming the subset is close to a metric.

\begin{lemma}
\label{lem:HST2tree}
Consider a semi-metric $\M=(X,\xi)$ of size $n$ in
$\bigoplus_{\min}^{l} \TM$ for some $l\in \N$, where each tree metric
in the product is of size $O(n)$. Suppose $\M$ is a $\gamma$-near
metric (i.e., it is embeddable into a metric with $\gamma$
distortion). Then we can embed $\M$ into a connected weighted graph
with $O(nl)$ edges with distortion $\gamma$ in $O(nl)$ time.
\end{lemma}

\begin{proof}
We consider $l$ separate trees each on $O(n)$ nodes, corresponding to
each of $l$ dimensions of the min-product.  We identify the nodes of
trees that correspond to the same point in the min-product, and
collapse them into a single node.  The graph we obtain has at most
$O(nl)$ edges. Denote the shortest-path metric it spans with
$\M'=(V,\rho)$, and denote our embedding with $\phi:X \to V$.
Clearly, for each pair $u,v$ of points in $X$, we have
$\rho(\phi(u),\phi(v)) \le \xi(u,v)$. If the distance between two
points shrinks after embedding, then there is a sequence of points $w_0=u$,
$w_1$, \ldots, $w_{k-1}$, $w_k=v$ such that $\rho(\phi(u),\phi(v)) =
\xi(w_0,w_1) + \xi(w_1,w_2) + \cdots + \xi(w_{k-1},w_k)$.  Because
$\M$ is a $\gamma$-near metric, there exists a metric $\xi^{\star}:X
\times X \to [0,\infty)$, such that $ \xi^\star(x,y) \le \xi(x,y) \le
  \gamma \cdot \xi^\star(x,y)$, for all $x,y\in X$.  Therefore,
\begin{equation*}
\rho(\phi(u),\phi(v)) = \sum_{i=0}^{k-1} \xi(w_i,w_{i+1})
\ge \sum_{i=0}^{k-1}\xi^\star(w_i,w_{i+1})
\ge \xi^\star(w_0,w_k) =  \xi^\star(u,v) \ge \xi(u,v) / \gamma.
\end{equation*}
Hence, it suffices to multiply all edge weights of the graph by
$\gamma$ to achieve a non-contractive embedding. Since there was no expansion
before, it is now bounded by $\gamma$.
\end{proof}

We now show how to embed the shortest-path metric of a graph into a
low dimensional $\ell_1$-space in time near-linear in the graph size.
For this purpose, we implement Bourgain's embedding~\cite{Bou} in
near-linear time. We use the following version of Bourgain's
embedding, which follows from the analysis in~\cite{Matousek-book}.

\begin{lemma}[Bourgain's embedding~\cite{Matousek-book}]
\label{lemma:bourgain}
Let $\M=(X,\rho)$ be a finite metric on $n$ points.  There is
an algorithm that computes an embedding $f:X \to \ell_1^t$ of
$\mathcal M$ into $\ell_1^t$ for $t=O(\log^2 n)$ such that, with high
probability, for each $u,v \in X$, we have $\rho(u,v)\le\| f(u) - f(v)
\|_1 \le \rho(u,v) \cdot O(\log n)$. 

Specifically, for coordinate $i\in[k]$ of $f$, the embedding associates
a nonempty set $A_i\subseteq X$ such that $f(u)_i=\rho(u,A_i)=\min_{a\in
  A_i}\rho(u,a)$. Each $A_i$ is samplable in linear time.

The running time of the algorithm is $O(g(n) \cdot \log^2 n)$, where
$g(n)$ is the time necessary to compute the distance of all points to
a given fixed subset of points.
\end{lemma}

\begin{lemma}
\label{lem:tree2Bourgain}
Consider a connected graph $G=(V,E)$ on $n$ nodes with $m$ edges and a
weight function $w:E \to [0,\infty)$.  There is a randomized algorithm
that embeds the shortest path metric of $G$
into $\ell_1^{O(\log^2 n)}$ with $O(\log n)$ distortion,
with high probability, in $O(m\log^3n)$ time.
\end{lemma}

\begin{proof}
Let $\psi:V \to \ell_1^{O(\log^2 n)}$ be the embedding given by
Lemma~\ref{lemma:bourgain}. 
For any nonempty subset $A\subseteq V$, we can compute
$\rho(v,A)$ for all $v\in V$ by Dijkstra's algorithm in $O(m\log n)$ time.
The total running time is thus $O(m\log^3 n)$.
\end{proof}

\subsubsection{Finalization of the proof of Theorem~\ref{thm:streamEMD}}
\label{sec:pfStEMD}

We first apply Lemma~\ref{lem:slidingWindow} to embed the sets $A_i$
into $\bigoplus_{\rm min}^{O(\log n)}\ell_1^{k}$ with distortion at
most $O(\log^2 n)$ with high probability, where $k = O(\log^3 n)$. We
write $v_i$, $i\in[n-s+1]$, to denote the embedding of $A_i$. Note
that the TEMD distance between two different $A_i$'s is at least $1/s
\ge 1/n$, and so is the distance between two different $v_i$'s.  We
multiply all coordinates of $v_i$'s by $2kn = \tO(n)$ and round them
to the nearest integer.  This way we obtain vectors $v'_i$ with
integer coordinates in $\{-2knM-1,\ldots,2knM+1\}$.  Consider two
vectors $v_i$ and $v_j$. Let $D$ be their distance, and let $D'$ be
the distance between the corresponding $v'_i$ and $v'_j$. We claim
that $knD \le D' \le 3knD$, and it suffices to show this claim for
$v_i \ne v_j$, in which case we know that $D \ge 1/n$. Each coordinate
of the min-product is $\ell_1^k$, and we know that in each of the
coordinates the distance is at least $D$. Consider a given coordinate
of the min-product, and let $d$ and $d'$ be the distance before and after
the scaling and rounding, respectively. On the one hand,
$$\frac{d'}{d} \ge \frac{2knd - k}{d} \ge 2kn - \frac{k}{D} \ge 2kn - kn = kn,$$
and on the other,
$$\frac{d'}{d} \le \frac{2knd + k}{d} \le 2kn + \frac{k}{D} \le 2kn + kn = 3kn.$$
Therefore, in each coordinate, the distance gets scaled by a factor in the range $[kn,3kn]$.
We now apply Lemma~\ref{lem:l1ToHST} to $v'_i$'s and obtain their embedding into a min-product of tree metrics. Then, we divide all distances in the trees by $kn$, and achieve an embedding of $v_i$'s into a min-product of trees with distortion at most 3 times larger than that implied by Lemma~\ref{lem:l1ToHST}, which is $O(\log n)$.

The resulting min-product of tree metrics need not be a metric, but it
is a $\gamma$-near metric, where $\gamma=O(\log^3 n)$ is the expansion
incurred so far. We therefore embed the min-product of tree metrics
into the shortest-path metric of a weighted graph by using
Lemma~\ref{lem:HST2tree} with expansion at most $\gamma$. Finally, we
embed this metric into a low dimensional $\ell_1$ metric space with
distortion $O(\log^2 n)$ by using Lemma~\ref{lem:tree2Bourgain}.

\section{Applications}

We now present two applications mentioned in the introduction:
sublinear-time approximation of edit distance, and approximate pattern matching
under edit distance.

\subsection{Sublinear-time approximation}\label{appx:sublinear}
We now present a sublinear-time algorithm for distinguishing pairs of
strings with small edit distance from pairs with large edit
distance. Let $x$ and $y$ be the two strings.  The algorithm
partitions them into blocks $\widetilde x_i$ and $\widetilde y_i$ of
the same length such that $x=\widetilde x_1\ldots \widetilde x_b$ and
$y=\widetilde y_1\ldots \widetilde y_b$. Then it selects a few random $i$, and for each of them, it compares $\widetilde x_i$ to $\widetilde y_i$.
If it finds an $i$ for which $\widetilde x_i$ and $\widetilde y_i$
are very different, the distance between $x$ and $y$ is likely to be large.
Otherwise, if no such $i$ is detected, the edit distance between $x$ and $y$ is likely to be small.
Our edit distance algorithm is used for approximating the distance between specific
$\widetilde x_i$ and $\widetilde y_i$.

\begin{theorem}
Let $\alpha$ and $\beta$ be two constants such that $0 \le \alpha < \beta \le 1$.
There is an algorithm that distinguishes pairs of strings with edit distance $O(n^\alpha)$
from those with distance $\Omega(n^\beta)$ in time $n^{\alpha + 2(1-\beta) + o(1)}$.
\end{theorem}

\begin{proof}
Let $f(n) = 2^{O(\sqrt{\log n\log\log n})}$ be a non-decreasing function that
bounds the approximation factor of the algorithm given by
Theorem~\ref{thm:main}. Let $b = \frac{n^{\beta - \alpha}}{f(n) \cdot
  \log n}$. We partition the input strings $x$ and $y$ into $b$ blocks,
denoted $\widetilde x_i$ and $\widetilde y_i$ for $i\in[b]$, of length $n/b$ each. 

If $\ed(x,y) = O(n^\alpha)$, then $\max_i \ed(\widetilde x_i,\widetilde y_i) \le \ed(x,y) =
O(n^\alpha)$.  On the other hand, if $\ed(x,y) = \Omega(n^\beta)$,
then $\max_i  \ed(\widetilde x_i,\widetilde y_i) \ge \ed(x,y) / b =
\Omega(n^\alpha \cdot f(n) \cdot \log n)$.  Moreover, the number of
blocks $i$ such that $\ed(\widetilde x_i,\widetilde y_i) \ge \ed(x,y) / 2b = \Omega(n^\alpha
\cdot f(n) \cdot \log n)$ is at least $$ \frac{\ed(x,y) - b \cdot
  \ed(x,y) / 2b}{n/b} = \Omega(n^{\beta -1} \cdot b).$$ Therefore, we
can tell the two cases apart with constant probability by sampling
$O(n^{1-\beta})$ pairs of blocks $(\widetilde x_i, \widetilde y_i)$ and checking if any of
the pairs is at distance $\Omega(n^\alpha \cdot f(n) \cdot \log n)$.
Since for each such pair of strings, we only have to tell edit
distance $O(n^\alpha)$ from $\Omega(n^\alpha \cdot f(n) \cdot \log
n)$, we can use the algorithm of Theorem~\ref{thm:main}. We amplify
the probability of success of that algorithm in the standard way by
running it $O(\log n)$ times. The total running time of the algorithm
is $O(n^{1-\beta}) \cdot O(\log n) \cdot (n/b)^{1+o(1)} = O(n^{\alpha
  + 2(1-\beta) + o(1)})$.
\end{proof}

\subsection{Pattern matching}
Our algorithm can be used for approximating the edit distance between
a pattern $P$ of length $n$ and all length-$n$ substrings of a text
$T$. Let $N = |T|$. For every $s\in[N-2n+1]$ of the form $in+1$, we
concatenate $T$'s length-$2n$ substring that starts at index $s$ with
$P$, and compute an embedding of edit distance between all length-$n$
substrings of the newly created string into $\ell_1^{\alpha}$ for
$\alpha=2^{O(\sqrt{\log n\log \log n})}$. We routinely amplify
the probability of success of each execution of the algorithm by
running it $O(\log N)$ times and selecting the median of the
returned values.  The running time of the algorithm is $O(N \log N)
\cdot 2^{O(\sqrt{\log n\log\log n})}$.

The distance between each of the substrings and the pattern is approximate up
to a factor of $2^{O(\sqrt{\log n \log\log n})}$, and can be used both for
finding approximate occurrences of $P$ in $T$, and for finding a substring of $T$
that is approximately closest to $P$.

\section*{Acknowledgment}

The authors thank Piotr Indyk for helpful discussions, and Robert
Krauthgamer, Sofya Raskhodnikova, Ronitt Rubinfeld, and Rahul Sami for
early discussions on near-linear algorithms for edit distance.

{\small
\bibliographystyle{alpha}
\bibliography{bibfile}
}

\end{document}